\documentclass[12pt,twoside]{article}
\usepackage[mathscr]{eucal}
\usepackage{amsthm,amsmath,amssymb,amscd}
\usepackage{graphicx}
\usepackage{latexsym,enumerate}
\usepackage{color}
\usepackage{tikz}
\usetikzlibrary{shapes.arrows,chains,positioning}
\usepackage[outline]{contour}
\usepackage{caption}
\usepackage{subcaption}
\newtheorem{theorem}{Theorem}
\newtheorem{lemma}{Lemma}
\newtheorem{proposition}{Proposition}

\newtheorem{condition}{Planarity Condition}

\theoremstyle{definition}

\newtheorem*{remark}{Remark}

\usepackage{color,hyperref}
\definecolor{dark-blue}{rgb}{0.15,0.15,0.4}
\definecolor{dark-red}{rgb}{0.4,0.15,0.15}
\definecolor{medium-red}{rgb}{0.6,0,0}
\definecolor{medium-blue}{rgb}{0,0,0.6}
\hypersetup{
    colorlinks=true,
    linktoc=all,
    citecolor=medium-red,
    filecolor=medium-blue,
    linkcolor=medium-blue,
    urlcolor=medium-blue,
}

\begin{document}
\title{Four-body central configurations with one pair of opposite sides parallel}
\author{Manuele Santoprete\thanks{ Department of Mathematics, Wilfrid Laurier
University E-mail: msantopr@wlu.ca}} 
 \maketitle

\begin{abstract}
    We study  four-body central configurations with one pair of opposite sides parallel. We use a novel constraint to write the  central configuration equations in this special case, using distances as variables.  We prove that, for a given ordering of the mutual distances, a trapezoidal central configuration must have a  certain partial ordering of the masses. We also  show that if opposite masses of a four-body trapezoidal central configuration are equal, then the configuration has a line of symmetry and it must be a  kite. In contrast to the general four-body case, we show that if  the  two adjacent masses bounding the shortest side  are equal, then the configuration must be  an isosceles trapezoid, and the remaining two masses must also be equal.
  \end{abstract}

\renewcommand{\thefootnote}{\alph{footnote})}
\tableofcontents

\section{Introduction}
Let $ P _1,P _2, P _3$, and $P _4 $ be  four points in $\mathbb{R}^3$ with position vectors $\mathbf{q} _1 , \mathbf{q} _2 , \mathbf{q} _3  $, and $ \mathbf{q} _4 $, respectively. Let $ r _{ ij } = \| \mathbf{q} _i - \mathbf{q} _j \| $, be the distance between the point $ P _i $ and $ P _j $, and let $ \mathbf{q} = (\mathbf{q} _1 , \mathbf{q} _2 , \mathbf{q} _3 , \mathbf{q} _4) \in \mathbb{R}  ^{ 12 }$.
The center of mass of the system is $ \mathbf{q} _{ CM } = \frac{ 1 } { M} \sum _{ i = 1 } ^n m _i \mathbf{q} _i $, where $ M = m _1 + \ldots m _n $ is the total mass. 
The Newtonian $4$-body problem concerns the motion of $4$ particles
with masses $m_i\in{\mathbb R}^+$ and positions $\mathbf{q} _i\in{\mathbb
R}^3$, where $i=1,\ldots,4$. The motion is governed by Newton's
law of motion 
\begin{equation} 
 m _i\mathbf{\ddot q} _i=   \sum _{ i \neq j } \frac{m _i m _j (\mathbf{q} _j - \mathbf{q} _i)   } { r _{ ij } ^3 }=\frac{\partial  U}{  \partial  \mathbf{q}  _i}, \quad 1\leq i\leq 4
\end{equation} 
where $U(\mathbf{q} )$ is the Newtonian potential
\begin{equation} 
U( \mathbf{q} )=\sum_{i<j}\frac{m_im_j}{r_{ij}},\quad 1\leq i\leq 4.
\end{equation} 
A {\it central configuration} (c.c.) of the four-body problem is a configuration $ \mathbf{q} \in \mathbb{R} ^{ 12} $ which satisfies the algebraic equations
\begin{equation} \label{eqn:cc1}
 \lambda\, m _i (\mathbf{q} _i -\mathbf{q} _{ CM })  = \sum _{ i \neq j } \frac{ m _i m _j (\mathbf{q} _j - \mathbf{q} _i) } { r _{ ij } ^3 }, \quad 1 \leq i \leq n .
 \end{equation} 
If we let $ I (\mathbf{q}) $ denote the moment of inertia, that is,
\[ I (\mathbf{q}) = \frac{1}{2} \sum _{ i = 1 } ^n m _i \| \mathbf{q} _i - \mathbf{q} _{ CM } \| ^2 = \frac{ 1}{2 M} \sum _{ 1 \leq i < j \leq n } ^n  m _i m _j r _{ ij } ^2, 
\]
    we can write equations \eqref{eqn:cc1} as 
    \begin{equation}\label{eqn:cc3}  \nabla U (\mathbf{q}) = \lambda \nabla I (\mathbf{q}).  
    \end{equation}
   Viewing $ \lambda $  as a Lagrange multiplier, a central configuration is simply a critical point of $U $
subject to the constraint $I$ equals a constant.

A central configuration is {\it planar} if the four points $ P _1,P _2, P _3$, and $P _4 $  lie on the same plane. Equations \eqref{eqn:cc1}, and \eqref{eqn:cc3}  also describe planar central configurations provided $ \mathbf{q} _i \in \mathbb{R}^2  $ for $ i = 1, \ldots 4 $. We say that a planar configuration is {\it degenerate} if two or more points coincide, or if more than two points lie on the same line. 
Non-degenerate planar configurations can be classified as either {\it concave} or {\it convex}. A concave configuration has one point which is located strictly inside the convex hull of the other three, whereas a convex configuration does not have a point contained in the convex hull of the other three points. Any  convex configuration determines a convex quadrilateral (for a precise definition  of quadrilateral see for example \cite{behnke1974fundamentals}). 
In a planar convex configuration we say that the points are {\it ordered sequentially} if they are numbered consecutively while traversing the boundary of the corresponding convex  quadrilateral.
In this paper we are interested in studying {\it trapezoidal central configurations}, that is, those c.c.'s for which  two of the opposite sides are parallel  (see Figure \ref{fig:trapez}). Non-degenerate trapezoidal central configurations are necessarily convex.  
 \begin{figure}
     \begin{center}
      \scalebox{.8}{
 \begin{tikzpicture}

\draw (-6,0) -- (6,0);
\draw (-6,5.5) -- (6,5.5);

\fill (4,0) circle (0.1cm);
\fill (-4,0) circle (0.1cm);
\fill (-3.,5.5) circle (0.1cm);
\fill (2,5.5) circle (0.1cm);

\draw[very thick] (-4,0)--(4,0);
\draw[very thick] (-4,0)--(-3,5.5);
\draw[very thick] (-3,5.5)--(2,5.5);
\draw[very thick] (-4,0)--(-3,5.5);
\draw[very thick] (2,5.5)--(4,0);

\node[] at (4,-0.4) {$m _1 $};
\node[] at (-4,-0.4) {$m _2 $};
\node[] at (-3,5.8) {$m _3 $};
\node[] at (2,5.8) {$m _4 $};

\end{tikzpicture}
}\end{center}
\caption{An example of a trapezoidal   central configuration.\label{fig:trapez}}
\end{figure}
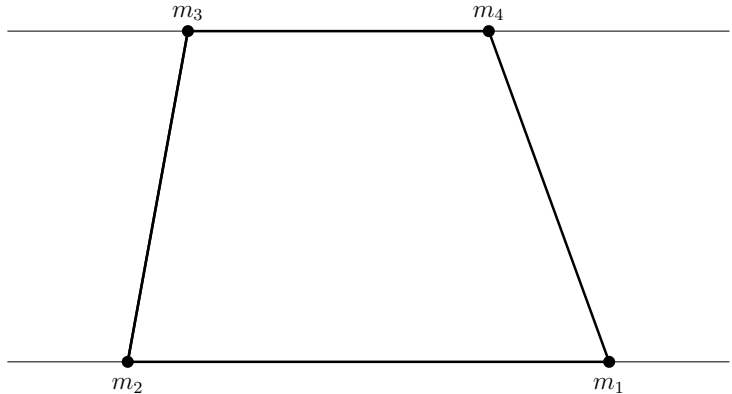


The four body problem has a  long and distinguished history. In 1900 Dziobek derived equations for central configurations of four bodies with distances as variables \cite{dziobek1900uber}. In 1932 McMillan and Bartky  used similar equations to obtain many important new results  \cite{macmillan1932permanent}. In 1996 Albouy \cite{albouy1995symetrie,albouy1996symmetric} gave a complete classifications of the four-body c.c.'s with equal masses. More recently, in 2006  Hampton and Moeckel \cite{hampton2006finiteness} proved the finiteness of the number of c.c.'s. Other recent results of note, concerning  four-body c.c.'s in the case some of the masses are equal, were obtained by  the present author and Perez-Chavela,  \cite{perez2007convex}, Albouy, Fu and Sun \cite{albouy2008symmetry},  and Fernandes, Llibre and  Mello \cite{fernandes2017convex}.
Further   results for the four-body were recently attained by Cors and Roberts \cite{cors2012four}, Corbera, Cors and Roberts \cite{corbera2016four}, Deng,   Li and Zhang \cite{deng2017four,deng2017some} and Xie \cite{xie2012isosceles}, just to mention a few. Particularly important for this paper is the work of Cors and Roberts \cite{cors2012four} which inspired for the approach we take here. 
Additionally, certain bifurcations in the four-body problem, and several planarity conditions and their applications to four-body c.c.'s were obtained by the present author in  \cite{rusu2016bifurcations} and  \cite{santoprete2017planarity}, respectively. 

Let $ \mathbf{r}= (r _{ 12 } , r _{ 13 } , r _{ 14 } , r _{ 23 } , r _{ 24 } , r _{ 34 }) \in (\mathbb{R}  ^{ + }) ^{ 6 } $ be a vector of mutual distances. The conditions for which such vector determines a realizable configuration of four bodies in Euclidean space  
can be expressed by the   Cayley-Menger criterion, that we state below. The Cayley-Menger determinant of four points  $ P _1 , \ldots P _4 $   is
\[   H (\mathbf{r} ) = \begin{vmatrix}
            0 & 1 & 1 & 1 & 1 \\
            1 & 0 & r^2 _{ 12 } & r^2 _{ 13 } & r^2 _{ 14 }  \\
            1 & r ^2_{ 12 } & 0 & r^2 _{ 23 } & r^2 _{ 24 } \\
            1 & r ^2_{ 13 } & r^2 _{ 23 } & 0 & r^2 _{ 34 } \\
            1 & r^2 _{ 14 } & r^2 _{ 24 } & r^2 _{ 34 } & 0.
        \end{vmatrix}. 
\]
A configuration is geometrically realizable if and only if the Cayley-Menger determinant of each subconfiguration of two or more points is $ \geq 0 $ when the number of points is even, and $ \leq 0 $ when it is odd.  See, for instance, the book of  Blumenthal  \cite{blumenthal_theory_1970} or  Theorem 9.7.3.4 and Exercise 9.14.23 in \cite{berger2009geometry}. Note that an  equivalent characterization can be given in terms of  Borchardt's quadratic form, see  \cite{albouy2006mutual,schoenberg_metric_1938}. 
 In the remainder of this paper we assume that $ \mathbf{r} $ is geometrically realizable. 

In the four-body problem the mutual distances are not independent  so that  describing  planar four-body central configurations requires an additional constraint. Following Dziobek \cite{dziobek1900uber} it is customary to use the following   planarity condition,
\begin{condition}
 $ P _1,P _2, P _3, P _4 \in \mathbb{R}^3$  are coplanar  if and only if the Cayley-Menger determinant  determined by  these four points is $ 0 $, that is, $ H (\mathbf{r}) = 0 $.
\end{condition}
In this paper we  use a different constraint that not only gives planarity of the configuration, but also restricts the configuration to be trapezoidal. This planarity conditions complements the list given in \cite{santoprete2017planarity}. Our approach parallels the treatment of the co-circular for body problem given by Cors and Roberts in \cite{cors2012four}. The new constraint is introduced in Section 2.
In Section 3 we derive the equations for the trapezoidal central configurations. 
In Section 4 we study the relationship between the Cayley-Menger constraint and the constraint used in the paper, and show that, as expected, the gradients of these restrictions are collinear at trapezoidal configurations. 
In Section 5 we prove that, for a given ordering of the mutual distances, a trapezoidal central configuration must have a  certain partial ordering of the masses. This result is by necessity weaker than the analogous result for co-circular configurations where one obtains a total ordering (see \cite{cors2012four}). We also  prove that if opposite masses of a four-body trapezoidal central configuration are equal, then the configuration has a line of symmetry and is a kite. This is a special case of the well known result of Albouy, Fu and Sun \cite{albouy2008symmetry}. A similar result also holds in  the case the  two adjacent masses bounding the shortest side  are equal. In this case the configuration is an isosceles trapezoid, and the remaining two masses must also be equal. Finally, we show that, in contrast to the co-circular case, when    the  two adjacent masses bounding the longest side  are equal there are asymmetric solutions. 
\section{Another Planarity Conditions}
Let $ P _1,P _2, P _3$, and $P _4 $ be  four points in $\mathbb{R}^3$ and let  $\mathbf{q} _1 , \mathbf{q} _2 , \mathbf{q} _3  $, and $ \mathbf{q} _4 $ be their position vectors. 
In this section we introduce a planarity condition that also  constrains the configuration to have one pair of
opposite sides parallel.
Let  
\begin{align*} 
    \mathbf{a}&  = \mathbf{q}  _2 - \mathbf{q}  _1 ,~ 
    \mathbf{b}  = \mathbf{q}  _3 - \mathbf{q}  _2 , ~
     \mathbf{c}   = \mathbf{q}  _4 - \mathbf{q}  _3 , \\
     \mathbf{d}  & = \mathbf{q}  _1 - \mathbf{q}  _4 , ~ 
     \mathbf{e}  = \mathbf{q}  _3 - \mathbf{q}  _1 , ~
    \mathbf{f}  = \mathbf{q}  _4 - \mathbf{q}  _2,    
\end{align*} 
then it follows that  $ \mathbf{a} + \mathbf{b} + \mathbf{c} + \mathbf{d}= 0 $,  $ \mathbf{f} = \mathbf{b} + \mathbf{c} $, and $ \mathbf{e} = \mathbf{a} + \mathbf{b} $, see figure \ref{fig:tetrahedron}. For convenience, we will also use  $a, b,c ,d, e, f $ to denote the mutual distances:
\[ a = r _{ 12 } ,\quad b = r _{ 23 }, \quad c = r _{ 34 }, \quad d = r _{ 14 },\quad  e = r _{ 13 } ,\quad f = r _{ 24 }.\]
\begin{figure}[t]\begin{center}
\begin{tikzpicture}[thick]
\clip (-2.5,-2.5) rectangle (3.5,3.6);

\draw[very thick,color=black,>=latex,->] (1,3) --   (-2,0)node[midway,above] {$ \mathbf{a}$  };
\draw[very thick,color=black,>=latex,->] (-2,0) -- (0,-2) node[midway,below] {$\mathbf{b}$ };
\draw[very thick,color=black,>=latex,<-] (3,0) -- (0,-2) node[midway,below] {$\mathbf{c}$ };

\draw[very thick,color=black,>=latex,<-] (1,3) -- (3,0) node[midway,above] {$\mathbf{d}$ };
\draw[very thick,color=black,>=latex,->] (1,3) -- (0,-2) node[midway,left] {$\mathbf{e}$ };
\draw[very thick,dashed,color=black,>=latex,->] (-2,0) -- (2.9,0) node[midway,right] {$\mathbf{f}$ };
\draw[color=black,fill=black] (1,3) circle (0.07)  node[above] {$ P _1 $ };
\draw[color=black,fill=black] (-2,0) circle (0.07)  node[below] {$ P_2 $ };
\draw[color=black,fill=black] (0,-2) circle (0.07)  node[below] {$ P _3 $ };
\draw[color=black,fill=black] (3,0) circle (0.07)  node[below] {$ P _4 $ };
\end{tikzpicture}
\end{center}
\caption{The points $ P _1, P _2 , P _3$ , and $P _4$  form a tetrahedron in $ \mathbb{R}  ^3 $.\label{fig:tetrahedron}}
\end{figure}
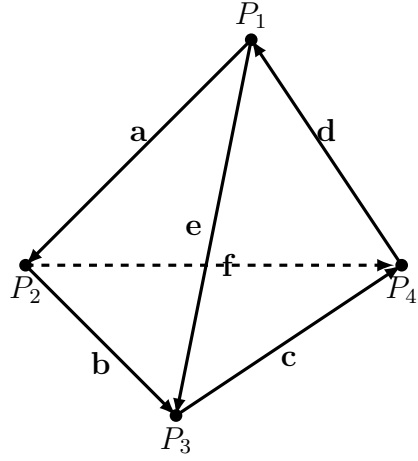

In the following lemma we introduce the quantity $ \Delta $ that will be shown to be of great significance for this work.
\begin{lemma}\label{lemma:ef}
    Let $ \Delta = \frac{1}{2} \| \mathbf{a} \times \mathbf{c} \| $, then, with the above definitions, the following equation holds
    \[
        4\Delta  ^2 = a ^2 c ^2 - \frac{1}{4} (b ^2 + d ^2 - e ^2 - f ^2 ) ^2.    \]
\end{lemma} 
\begin{proof} 
 Clearly,
\[
    4\Delta ^2  = (\mathbf{a} \times \mathbf{c}) \cdot (\mathbf{a} \times \mathbf{c}) =  (\mathbf{a} \cdot \mathbf{a }) (\mathbf{c} \cdot \mathbf{c}) -  (\mathbf{a} \cdot \mathbf{c} ) ^2 = a ^2 c ^2 -  (\mathbf{a} \cdot \mathbf{c}) ^2.       
\]
But 
\begin{align*}
   2 (\mathbf{a} \cdot \mathbf{c}) & = 2\cdot (\mathbf{f} + \mathbf{d} ) \cdot (\mathbf{d} + \mathbf{e})  =  2 \mathbf{d} \cdot (\mathbf{f} + \mathbf{d}) + 2 \mathbf{e} \cdot (\mathbf{f} + \mathbf{d}) \\ 
   & = 2 
   \mathbf{d}\cdot  (\mathbf{b} - \mathbf{e}) + 2 \mathbf{e} \cdot (\mathbf{f} + \mathbf{d}) = 2\, \mathbf{b} \cdot \mathbf{d} - 2 \,\mathbf{e} \cdot \mathbf{f}\\   
   &= (\mathbf{d} + \mathbf{b}) \cdot (\mathbf{d} + \mathbf{b}) - \mathbf{d} \cdot \mathbf{d} - \mathbf{b} \cdot \mathbf{b} - (\mathbf{e} - \mathbf{f}) \cdot (\mathbf{e} - \mathbf{f}) + \mathbf{e} \cdot \mathbf{e} + \mathbf{f} \cdot \mathbf{f} \\
   & = e ^2 + f ^2 - d ^2 - b ^2 + (\mathbf{d} + \mathbf{b})\cdot  (\mathbf{d} + \mathbf{b}) - ( \mathbf{b} + 2 \mathbf{c} + \mathbf{d} ^2 \\
& = e ^2 + f ^2 - d ^2 - b ^2 - \mathbf{c}  \cdot (\mathbf{b} + \mathbf{d}) - 4 \mathbf{c} \cdot \mathbf{c} \\
& = e ^2 + f ^2 - d ^2 - b ^2 - 4 c ^2 + 4 \mathbf{c} \cdot (\mathbf{a} + \mathbf{c}) \\
 & = e ^2 + f ^2 - d ^2 - b ^2 + 4 \mathbf{a} \cdot \mathbf{c}.
\end{align*}  
It follows that 
\[
    2 (\mathbf{a} \cdot \mathbf{c}) = d ^2 + b ^2 - e ^2 - f ^2.  
\]
Hence, 
\[
    4\,\Delta  ^2 = a ^2  c ^2 - \frac{1}{4} (e ^2 + f ^2 - b ^2 - d ^2 ) ^2.  
\] 
\end{proof}

In the case of a planar configuration $ \Delta $ can be interpreted as the absolute value of the 
difference of the areas of the triangles whose bases are the sides $ \mathbf{b} $ and $\mathbf{d}  $ of a convex  quadrilateral, and whose vertices coincide with the intersection of the diagonals (see \cite{hobson2004treatise} page 208). Note that $ \Delta $ can also be viewed as the area of a crossed quadrilateral (see \cite{coxeter1967geometry}).

There are two ways to obtain a planarity condition from this. One is to impose that $ \Delta $ is equal to the absolute value of the difference of the areas  $ A _3 $  and $ A _4 $ (or the absolute value of the  difference between $ A _1 $ and $ A _2 $). Here $ A _i $ is the area of the triangle whose vertices contain all bodies except for the $ i $-th body. 
The second approach, which is the one we take here, is to impose that $ \Delta = 0 $

\begin{condition}\label{cond:planarity1}
    Suppose $ \Delta = \frac{1}{2} \| \mathbf{a} \times \mathbf{c} \|$. Then,
    $ \Delta = 0 $ if and only if $ \mathbf{a} $ and $ \mathbf{c} $ are parallel and the configuration is planar. 
\end{condition}
\begin{proof}
   Clearly, if $ \| \mathbf{a} \times \mathbf{c} \|  = 0 $  the vectors $ \mathbf{a} $ and $ \mathbf{c} $ are parallel, in which case the configuration is planar because the four points lie on two parallel lines. 
Conversely, if the configuration is planar with $ \mathbf{a} $ and $ \mathbf{c} $ parallel, then $ \| \mathbf{a} \times \mathbf{c} \|  = 0 $. 
\end{proof} 

Note that the above condition can be written explicitly in terms of mutual distances as $a ^2  c ^2 = \frac{1}{4} (e ^2 + f ^2 - b ^2 - d ^2 ) ^2 $, or 
\begin{equation} \label{eqn:trapezoid0}   (2ac+e^2+f ^2-b^2-d^2) (2ac-e^2-f ^2+b^2+d^2) =0.     
\end{equation} 
 For the remainder of this paper we will assume that any trapezoidal configuration  satisfying the planarity condition above is  ordered sequentially so that $ r _{ 12 } , r _{ 3 4 } $ are the lengths of the bases of the trapezoid, $ r _{ 23 } $ and $ r _{ 14 }  $ are the lengths of legs,   and $ r _{ 13 } $ and $ r _{ 24 } $ are the lengths of the diagonals. In this case one has 
 \begin{equation}\label{eqn:trapezoid}
     (2ac-e^2-f ^2+b^2+d^2) =0,     
\end{equation} 
which  is known as  a necessary and sufficient condition for a convex quadrilateral with consecutive sides
 $ a,b,c,d $ and diagonals $e,f $  to be  a trapezoid with parallel sides $a $ and $c $. See for example \cite{josefsson2013characterizations}.

To double check that for realizable configurations  equation  \eqref{eqn:trapezoid} implies planarity, we proceed as follows. Substituting 
$ e^2=2ac-f ^2+b^2+d^2$ into the Cayley-Menger determinant yields 
\[-2(r_{12}^2r_{34}+r_{12}r_{23}^2-r_{12}r_{24}^2-r_{12}r_{34}^2-r_{14}^2r_{34}+r_{24}^2r_{34})^2 
    \leq 0.\]
By the Cayley-Menger criterion, this implies that for the mutual distances vector $ \mathbf{r} $ to correspond to a realizable configuration 
one must have 
\[r_{12}^2r_{34}+r_{12}r_{23}^2-r_{12}r_{24}^2-r_{12}r_{34}^2-r_{14}^2r_{34}+r_{24}^2r_{34}=0,\]
which in turn  implies that $ H (\mathbf{r}) = 0 $, and  leads to the formulas \eqref{eqn:diagonal1} and \eqref{eqn:diagonal2} for the diagonals of a trapezoid.

We remark that if one imposes  Ptolemy's condition  to study co-circular configurations, as done in \cite{cors2012four}, it is possible to see that  any realizable configuration satisfying Ptolemy's must be planar as a consequence of the   Cayley-Menger criterion. 
\section{Planarity Condition  and c.c equations }
In this section we give a derivation of the trapezoidal c.c.'s equations that mirrors the approach of Cors and Roberts \cite{cors2012four} for the co-circular problem. 
From  Planarity Condition \ref{cond:planarity1}, it follows that if we are looking for   planar central configurations  with opposite sides parallel, then  we can impose the  condition $ F = 4\Delta ^2 = 0 $.
Hence, we have the following proposition

\begin{proposition}  
    Assuming the bodies are sequentially ordered, a  trapezoidal central configuration is a critical point of the function 
    \begin{equation}
U+\lambda  M(I-I_0)+\sigma F.
\label{conditionSaari}
\end{equation}
satisfying $I-I_0=0$,  and  $ F =0 $, where  $\lambda $  and 
$ \sigma $ are Lagrange multipliers.
\end{proposition} 
Taking derivatives with respect to $ r _{ij}^2 $, and absorbing the $ \frac{1}{2} $ multiple into the Lagrange multiplier $ \sigma $,  we find that  the condition for a planar extrema is  
\begin{equation}\label{eqn:F1} 
m_im_j\left( \lambda -r_{ij}^{-3}\right)+ \sigma \frac{\partial F}{\partial r_{ij} ^2 }=0, \qquad  1\leq i<j\leq 4
\end{equation} 
\[I-I_0=0, \qquad F =0.\]

Writing \eqref{eqn:F1} explicitly yields 
\begin{align} 
m _1 m _2 ( r _{ 12 } ^{ - 3 } - \lambda)  & = \sigma\, r _{ 34 } ^2   &   m _3 m _4 ( r _{ 34 } ^{ - 3 } - \lambda)  & = \sigma\, r _{ 12 } ^2  \\ 
m _1 m _3 (r _{ 13 } ^{ - 3 } - \lambda) & = -\frac{1}{2} \sigma R &  m _2 m _4 (r _{ 24 } ^{ - 3 } - \lambda) & = -\frac{1}{2} \sigma R \\
m _1 m _4 (r _{ 14 } ^{ - 3 } - \lambda) & = \frac{1}{2} \sigma R &  m _2 m _3 (r _{ 23 } ^{ - 3 } - \lambda) & = \frac{1}{2} \sigma R,
\end{align}
where $ R = (r _{ 13 }  ^2 + r _{ 24 }  ^2-r _{ 14 }  ^2 - r _{ 23 }  ^2) $,
 together with $I-I_0=0$ and $F  =0$. Since $ F = 0 $, and we are assuming the ordering of the bodies described in the previous sections, then equation \eqref{eqn:trapezoid} is verified and hence it  follows that $ R = 2 r _{ 12 } r _{ 34 } $. Then, the  previous system of equations takes the form

\begin{align} 
m _1 m _2 ( r _{ 12 } ^{ - 3 } - \lambda)  & = \sigma\, r _{ 34 } ^2   &   m _3 m _4 ( r _{ 34 } ^{ - 3 } - \lambda)  & = \sigma\, r _{ 12 } ^2  \label{eqn:cc_n1}\\ 
m _1 m _3 (r _{ 13 } ^{ - 3 } - \lambda) & = - \sigma  r _{ 12 } r _{ 34 } &  m _2 m _4 (r _{ 24 } ^{ - 3 } - \lambda) & = - \sigma  r _{ 12 } r _{ 34 } \label{eqn:cc_n2}\\
m _1 m _4 (r _{ 14 } ^{ - 3 } - \lambda) & =  \sigma  r _{ 12 } r _{ 34 } &  m _2 m _3 (r _{ 23 } ^{ - 3 } - \lambda) & =  \sigma  r _{ 12 } r _{ 34 }, \label{eqn:cc_n3}
\end{align}

The equations have been grouped in pairs so that when they are multiplied together the product of the right-hand sides is  $  \sigma ^2 r _{ 34 } ^2 r _{ 12 } ^2 $. Consequently, the right hand sides are identical on the configurations satisfying  $ F = 0 $.
This yields the  well-known  relation of Dziobek \cite{dziobek1900uber}
\begin{equation}\label{eqn:dziobek}
   (r _{ 12 }^{ - 3 }  - \lambda) (r _{ 34 } ^{ - 3 } - \lambda) = (r _{ 13 } ^{  - 3 } - \lambda) (r _{ 24 } ^{ - 3 } - \lambda) = (r _{ 14 } ^{ - 3 } - \lambda) (r _{ 23 } ^{ - 3 } - \lambda),        
\end{equation} 
which is required of any planar 4-body central configuration (not only c.c.'s with parallel opposite sides). 

Eliminating $ \lambda $ from equation  \eqref{eqn:dziobek} and factoring gives the important relation
\begin{equation}\label{eqn:relation}
   (r _{ 13 } ^{ 3 } - r _{ 12 } ^3) (r _{ 23 } ^3  - r _{ 34 } ^3 ) (r _{ 24 } ^{ 3 } - r _{ 14 } ^3) = (r _{ 12 } ^3 - r _{ 14 } ^3) (r _{ 24 } ^3 - r _{ 34 } ^3) (r _{ 13 } ^3 - r _{ 23 } ^3).       
\end{equation} 
Assuming the six mutual distances determine an actual configuration in the plane,  this equation is necessary and sufficient for the existence of a  four-body planar central configuration. Further restrictions are needed to ensure that the masses are positive. 

Reasoning as in \cite{cors2012four} it is possible to show that  positivity of the masses implies that  each side of the quadrilateral is shorter in length than either diagonal, and that the shortest exterior side must lie opposite the longest. 
Then, the longest side will be either one of the parallel sides or one of the remaining exterior sides.
In the former case  suppose  $ r _{ 14 } $ is the longest exterior side, then we have that $ r _{ 23 } $ is the  shortest, and thus 
\[|r _{ 14 } - r _{ 23 } |> |r _{ 34 } - r _{ 12 } |.\] However, four lengths can constitute the consecutive sides of a non-parallelogram trapezoid, with $ r _{ 12 } $ and $ r _{ 34 } $ the lengths of the parallel sides, only when
\[|r _{ 14 } - r _{ 23 } |< |r _{ 34 } - r _{ 12 } |< r _{ 14 } + r _{ 23 },\]
which contradicts the previous inequality. 
A similar reasoning shows that $ r _{ 23 } $ cannot be the longest exterior side. 
Hence, in a trapezoidal central configuration, one of the legs cannot be the longest exterior side.

In the latter case, without any loss of generality, we can label the bodies so that  $ r _{ 12 } $   is the longest exterior side-length. Then, positivity of the masses implies that 
\[r _{ 13 } , r _{ 24 } > r _{ 12 } \geq r _{ 14 }, r _{ 23 } \geq r _{ 34 }.\]
With an appropriate relabeling it is also possible to assume $ r _{ 14 } \geq r _{ 23 } $ (see \cite{cors2012four}). This choice imposes   $ r _{ 13 } \geq r _{ 24 } $, and thus 
\begin{equation} \label{eqn:inequalities}
r _{ 13 } \geq r _{ 24 } > r _{ 12 } \geq r _{ 14 } \geq  r _{ 23 } \geq r _{ 34 }.
\end{equation} 
To prove the relation between the diagonals, recall that the lengths of the diagonals in a trapezoid are given by
(see \cite{josefsson2013characterizations}):
\begin{align}
r _{ 13 } = \sqrt{r _{ 12 } r _{ 34 } - \frac{ r _{ 34 } r _{ 23 } ^2 - r _{ 12 } r _{ 14 } ^2 }{r _{ 12 } - r _{ 34 } } }  \label{eqn:diagonal1}\\
r _{ 24 } = \sqrt{r _{ 12 } r _{ 34 } - \frac{ r _{ 34 } r _{ 14 } ^2 - r _{ 12 } r _{ 23 } ^2 }{r _{ 12 } - r _{ 34 } } }\label{eqn:diagonal2}.
\end{align} 
It follows that 
\[ r _{ 13 } ^2 - r _{ 24 } ^2 = \frac{ (r _{ 14 } ^2 - r _{ 23 } ^2) (r _{ 34 } + r _{ 12 }) } { r _{ 12 } - r _{ 34 }} \geq 0 ,   \]
since $ r _{ 12 } > r _{ 34 } $, and $ r _{ 14 } \geq r _{ 23 } $. 

Hence, without loss of generality we can restrict our analysis to the set 
\[
    \Omega = \{ \mathbf{r} \in (\mathbb{R}  ^+)^6 : r _{ 13 } \geq r _{ 24 } > r _{ 12 } \geq r _{ 14 } \geq  r _{ 23 } \geq r _{ 34 } \}.
\]

From the different ratios of two masses that can be derived from  equations(\ref{eqn:cc_n1}-\ref{eqn:cc_n3}), we obtain the following set of equations:
\begin{align}
    \frac{ m _1 } { m _2 } & = -\frac{ r _{ 23 } ^{ - 3 } - r _{ 24 } ^{ - 3 } } { r _{ 13 } ^{ - 3 } - r _{ 14 } ^{ - 3 } }  & \frac{ m _1 } { m _3 } & = \frac{ r _{ 34 } (r _{ 23 } ^{ - 3 } - r _{ 34 } ^{  - 3 })  }{ r _{ 12 } (r _{ 12 } ^{ - 3 } - r _{ 14 } ^{ - 3 }) }\label{eqn:masses1} \\
    \frac{ m _1 } { m _4 } & = -\frac{ r _{ 34} (r _{ 24 } ^{ - 3 } - r _{ 34 } ^{ - 3 }) } { r _{ 12 } (r _{ 12 } ^{ - 3 } - r _{ 13 } ^{ - 3 }) }  & \frac{ m _2} { m _3 } & = -\frac{ r _{ 34} (r _{ 13 } ^{ - 3 } - r _{ 34 } ^{ - 3 }) } { r _{ 12 } (r _{ 12 } ^{ - 3 } - r _{ 24 } ^{ - 3 }) } \label{eqn:masses2} \\ 
    \frac{ m _2} { m _4 } & = \frac{ r _{ 34} (r _{ 14 } ^{ - 3 } - r _{ 34 } ^{ - 3 }) } { r _{ 12 } (r _{ 12 } ^{ - 3 } - r _{ 23 } ^{ - 3 }) }  & \frac{ m _3} { m _4 } & = -\frac{ r _{ 14 } ^{ - 3 } - r _{ 24 } ^{ - 3 } } { r _{ 13 } ^{ - 3 } - r _{ 23 } ^{ - 3 } }.\label{eqn:masses3} 
\end{align} 
\section{Relationship to Cayley-Menger}
Let $ \Delta _i $ be the oriented area of the triangle whose vertices contain all bodies except for the
$ i$ -th body. For a quadrilateral ordered sequentially, we have $ \Delta _1 , \Delta _3 >0 $ and $ \Delta _2, \Delta _4<0 $. The derivatives of the Cayley-Menger determinant at planar c.c.'s are given by the following formula due to Dziobek \cite{dziobek1900uber}
\[ 
    \frac{ \partial H} { \partial r _{ ij } ^2  } (\mathbf{r})  = -32 \Delta _i \Delta _j.  
\]
In a trapezoid the areas $ |\Delta _i|$ take the form:
\[
    |\Delta _1| = |\Delta _2| =  \frac{1}{2}  r _{ 34 } h, \quad  |\Delta _3| = |\Delta _4| =  \frac{1}{2}  r _{ 12 } h
\]
where $h$ is the height, that is,  the distances between the opposite parallel sides. If the parallel sides have different lengths (i.e., $ r _{ 12 } \neq r _{ 34 } $) the  height of a trapezoid can be expressed in terms of mutual distances as follows:

\[
    h = \frac{ \sqrt{(a - c + d + e )(-a + c + d + e )(a + c -d + e )(a + c + d - e )}}{ 2 |a - c |}.
\]
If  $ r _{ 12 } = r _{ 34 } $, then the trapezoid reduces to a parallelogram, in which case, since the area is $ A = r _{ 12 } h $, we have 
\[h = \frac{ r _{ 12 } } A\]
and $ A $ is given by Bretschneider's formula for the  area of a quadrilateral, that is, 
\[
    A =\frac{1}{2} \sqrt{   e ^2  f ^2 - \frac{1}{4} (b ^2 + d ^2 - a ^2 - c ^2 ) ^2  }.
\]
In any case, since 
\[
     \frac{ \partial H} { \partial r _{ ij }  } (\mathbf{r})  = \frac{ \partial H} { \partial r _{ ij } ^2  } (\mathbf{r})\cdot \frac{ d (r _{ ij } ^2) } { d r _{ ij } }   = -64 r _{ ij } \, \Delta _i \Delta _j 
\]
we find that at a trapezoidal central configuration 
\[\nabla H (\mathbf{r}) = 16 r _{ 12 } r _{ 34 } h ^2 (r _{ 34 } , - r _{ 13 } , r _{ 14 } , r _{ 23 } , - r _{ 24 } , r _{ 12 }),\]
where $ h $ is defined above. On the other hand, the gradient of $ F (\mathbf{r}) $  at a trapezoidal configuration is,
\[\nabla F (\mathbf{r}) = 2 r _{ 12 } r _{ 34 }   (r _{ 34 } , - r _{ 13 } , r _{ 14 } , r _{ 23 } , - r _{ 24 } , r _{ 12 }) .\]
Comparing the two gradients above, we have the following proposition
\begin{proposition}
  For any trapezoidal central configuration $ \mathbf{r} $ 
  \[
      \nabla H ( \mathbf{r} ) = 8 h^2~  \nabla F ( \mathbf{r} ).  
  \]  
  In other words, on the set of geometrically realizable vectors for which both $H $ and $F$ vanish, the
gradients of these two functions are parallel.
\end{proposition} 
Note that a isosceles trapezoid central configurations is both a trapezoidal and co-circular. Therefore taking the proposition above together with Lemma 2.1 in \cite{cors2012four} implies that on the set of geometrically realizable vectors for which $H $, $F$   and $ P =  r _{ 12 } r _{ 34 } +r _{14 } r _{ 23 } - r _{ 13 } r _{ 24 } $  vanish the  gradients of these three functions are parallel. Thus, the codimension one level
surfaces defined by the equations $H=0 $, $F=0$,   and $ P=0 $    meet tangentially at the isosceles trapezoid configurations.

\section{Some Applications of the c.c. equations}
 In this section we apply equations \eqref{eqn:cc1}-\eqref{eqn:cc3} to obtain interesting results in the cases some of the masses are equal. We begin with two  propositions that have a very simple proof.
\begin{proposition}\label{prop:isosceles}
  If $ m _1 = m _2 $ and $ m _3 = m _4 $, then the corresponding trapezoidal central configuration
must be an isosceles trapezoid. 
\end{proposition} 
\begin{proof} 
By using the two equations \eqref{eqn:cc_n2} it follows that that the diagonals are equal, that is, $ r _{ 13 } = r _{ 24 } $. To conclude the proof it is enough to observe that a  trapezoid with diagonals of equal length is an isosceles trapezoid. 
\end{proof}

\begin{proposition}\label{prop:rhombus}
  If $ m _1 = m _3 $ and $ m _2 = m _4 $, then the corresponding trapezoidal central configuration
must be a rhombus. 
\end{proposition} 
\begin{proof} 
By using the two equations \eqref{eqn:cc_n3} it follows that  $ r _{ 14 } = r _{ 23 } $.
It follows that the quadrilateral is either a isosceles trapezoid or a rhombus.

If it is an isosceles trapezoid the diagonals have equal length, and from equation  \eqref{eqn:cc_n2} it follows that the masses are all equal. From the first of equation \eqref{eqn:masses2} it follows that the bases are equal, that is $ r _{ 12 } = r _{ 34 } $. Moreover, from equations \eqref{eqn:cc_n1} and \eqref{eqn:cc_n2}  we have that all the exterior sides are equal. The quadrilateral is then a square. 
In either case the quadrilateral is a rhombus.  

\end{proof} 

Cors and Roberts \cite{cors2012four} proved  that in any co-circular configuration with a given ordering of the mutual distances the masses must be ordered in a precise fashion, that is, the set of masses $ \{ m _1 , m _2 , m _3 , m _4 \} $ is totally ordered. 
We now want to obtain an similar  result in the case of trapezoidal central configurations. Because of the different geometry, however, it turns out that the set of masses is not totally ordered.  In fact, in  this case, it is only possible to obtain the following weaker result


\begin{theorem} \label{thm:masses}
Any trapezoid central configuration in $ \Omega $
satisfies
\[m _3 \leq m _4 \leq m _2\quad \mbox{and} \quad m _3 \leq m _1 . \]
\end{theorem} 

Before proving this we present two important inequalities needed in the proof of the theorem.  

\begin{lemma}\label{lem:decr_func}
    Let $ \phi : I \to \mathbb{R}  $ be a decreasing differentiable function on an interval $ I \subset \mathbb{R}  $. Suppose that $ x _1 \leq x _2 \leq x _3 \leq x _4 $, then 
    \[
        \frac{ \phi (x _2) - \phi (x _3) } { \phi (x _1) - \phi (x _4) } \leq 1.    
    \] 
\end{lemma}
\begin{proof} The proof is straightforward. Since $ x _1 \leq x _2 \leq x _3 \leq x _4 $ and $ \phi $ is decreasing, then $ \phi (x _1) \geq \phi (x _2) \geq \phi (x _3) \geq \phi (x _4) $. It follows that 
$ \phi (x _1) - \phi (x _4) \geq \phi (x _2) - \phi (x _3) $. Since both sides of the inequality are positive we obtain 
\[
    \frac{ \phi (x _2) - \phi (x _3) } { \phi (x _1) - \phi (x _4) } \leq 1,
\]
which is our claim. 
\end{proof} 

\begin{lemma}\label{lem:r3412} Any trapezoidal central configuration in $ \Omega $ satisfies 
    \[
         \frac{ r _{ 23} ^2} { r _{ 14 }^2 } 
         \geq \frac{ r _{ 34 } } { r _{ 12 } },  \]
     where the equality sign holds if, and only if, the trapezoid is a parallelogram with $ r _{ 12 } = r _{ 34 } $ and $ r _{ 23 } = r _{ 14 } $.
\end{lemma}
%
\begin{proof}
    If $ r _{ 12 } \neq r _{ 34 } $, then we can use equation \eqref{eqn:diagonal2}. Since   $ r _{ 24 } >r _{ 12 } $,  we have that
    \[ r _{ 12 } r  _{ 34 } (r _{ 12 } - r _{ 34 }) - r _{ 34 } r _{ 14 } ^2 + r _{ 12 } r _{ 23 } ^2 > r _{ 12 }^2 (r _{ 12 } - r _{ 34 }),  \]
        or $ - r _{ 34 } r _{ 14 } ^2 > r _{ 12 } (r _{ 34 } ^2 + r _{ 12 } ^2 - r _{ 23 } ^2 - 2 r _{ 12 } r _{ 34 }) $. Using condition \eqref{eqn:trapezoid}, that is, $ 2r _{ 12 } r _{ 34 } = r _{ 13 } ^2 + r _{ 24 } ^2 - r _{ 23 } ^2 - r _{ 14 } ^2 $ we obtain 
        \begin{align*}
           r _{ 34 } r _{ 14 } ^2&  < r _{ 12 } (r _{ 13 } ^2 + r _{ 24 } ^2 - r _{ 34 } ^2 - r _{ 12 } ^2 - r _{ 14 } ^2)    \leq r _{ 12 }  r _{ 23 } ^2, 
        \end{align*}   
where we used Euler's   quadrilateral inequality $ r _{ 13 } ^2 + r _{ 24 } ^2 \leq r _{ 12 } ^2 + r _{ 14 } ^2 + r _{ 23 } ^2 + r _{ 34 } ^2$ (see  \cite{dunham2000quadrilaterally}, or any good book on Euclidean geometry for a proof). It follows that $  \frac{ r _{ 23} ^2} { r _{ 14 }^2 } 
        > \frac{ r _{ 34 } } { r _{ 12 } } $.

        If $ r _{ 12 } = r _{ 34 } $ then the trapezoid degenerates to a parallelogram and thus $ r _{ 23 } = r _{ 14 } $. It follows that $1=  \frac{ r _{ 23} ^2} { r _{ 14 }^2 } 
        \geq  \frac{ r _{ 34 } } { r _{ 12 } } =1$, which completes the proof. 
\end{proof} 
We are now ready to prove  Theorem \ref{thm:masses}.
\begin{proof}[Proof of Theorem \ref{thm:masses}]
We first  prove that $ m _3 \leq m _4 $.
Using the second of equations \eqref{eqn:masses3} we have 
\[
    \frac{ m _3} { m _4 }  =  
    \frac{  r _{ 14 } ^{-3} - r _{ 24 } ^{-3} } { r _{ 23 } ^{-3} - r _{ 13 } ^{-3}}.
\]
 Let $ \phi: \mathbb{R}^+ \to \mathbb{R}  $ be defined by $ \phi (x) = x ^{ - 3 } $, then $ \phi $ is a decreasing function. An application of Lemma \ref{lem:decr_func} with $ x _1 = r _{ 23 }$, $x _2 = r _{ 14 } $, $ x _3 = r _{ 24 } $, and $ x _4 = r  _{ 13 } $, yields that 
 \[
     \frac{ \phi (x _2) - \phi (x _3) } { \phi (x _1) - \phi (x _4) }=\frac{  r _{ 14 } ^{-3} - r _{ 24 } ^{-3} } { r _{ 23 } ^{-3} - r _{ 13 } ^{-3}} \leq 1. 
 \] 
 It follows that,
\begin{equation}\label{eqn:m3m4} m _3 \leq m _4.  
\end{equation} 

Next, we verify that $ m _2 \geq m _4 $.
Multiplying the second equation in  \eqref{eqn:masses2} and the second equation in  \eqref{eqn:masses3}, we find 
\begin{align*}
    \frac{ m _2 } {  m _4 }  
    & =    \frac{ r _{ 12 } ^2 r _{ 23 } ^3 } { r _{ 34 } ^2  r _{ 14 } ^3 } \cdot 
   \frac{(r _{ 13 } ^3 - r _{ 34 } ^3) (r _{ 14 } ^3 - r _{ 24 } ^3)    }{ (r _{ 12 } ^3 - r _{ 24 } ^3)  (r _{ 13 } ^3 - r _{ 23 } ^3)    } \\
   & =\left( \frac{ r _{ 12 } r _{ 23 } ^2 } { r _{ 34 } r _{ 14 } ^2 } \right)  ^2 \cdot  \frac{ r _{ 14 } } { r _{ 23 } } \cdot  \frac{(r _{ 13 } ^3 - r _{ 34 } ^3) }{(r _{ 13 } ^3 - r _{ 23 } ^3)} \cdot \frac{(r _{ 24 } ^3 - r _{ 14 } ^3)    }{   (r _{ 24 } ^3 - r _{ 12 } ^3)    }.
\end{align*} 
All the fractions in the equation above are greater than or equal to one. In fact, the first fraction  is greater than or equal to one by Lemma \ref{lem:r3412},  the second fraction  because $ r _{ 14 } \geq r _{ 23 } $, and  the last two fractions because $ r _{ 23 } 
\geq r _{ 34 } $ and $ r _{ 12 } \geq r _{ 14 } $, respectively. This shows that $ m _2 \geq m _4 $.

Finally, using equations \eqref{eqn:masses1} and \eqref{eqn:relation} we find that 
\[
    \frac{ m _1 } { m _3 } = \frac{ r _{ 14 } ^3 } { r _{ 23 } ^3 } \cdot \frac{ r _{ 12 } ^2 } { r _{ 34 } ^2 } \cdot\frac{  r _{ 23 } ^3 - r _{ 34 } ^3 } { r _{ 12 } ^3 - r _{ 14 } ^3 }
    = \frac{ r _{ 14 } ^3 } { r _{ 23 } ^3 } \cdot \frac{ r _{ 12 } ^2 } { r _{ 34 } ^2 }\cdot \frac{ (r _{ 24 } ^3 - r _{ 34 } ^3)}{(r _{ 24 } ^3 - r _{ 14 } ^3) } \cdot \frac{(r _{ 13 } ^3 - r _{ 23 } ^3) } { (r _{ 13 } ^3 - r _{ 12 } ^3)  } .   
\]
All the fractions in the equation above are greater  than or equal to  one. To me more precise, the first two fractions are greater than  or equal to one because $ r _{ 14 } \geq r _{ 23 } $ and $ r _{ 12 } \geq r _{ 34 }  $, respectively. The last two fractions  because $ r _{ 14 } \geq r _{ 34 } $ and $ r _{ 12 } \geq r _{ 23 } $, respectively. Hence, $ m _1 \geq m _3 $, and the proof is complete.  

\end{proof} 

\begin{remark}
   From numerical experiments it appears  that $m _1 $ can be larger or smaller than either $ m _2 $ and $ m _4 $. 
  For example choosing 
 
  \begin{align*}    
      r_{13}&  = 9.7414781617108145730, &   r_{24} & = 8.75000000000000000, \\
       r_{12}  & = 8 & r_{14} & = 7.52080447824566090, \\
      r_{23} & = 7.1064329749865061893, & r_{34} & = 4.0246879466945716437, 
  \end{align*}
 we have $ m _1 / m _2 = 1.0194571510769873907 $ and $ m _1 / m _4 = 7.9942119368105807422 $, also these distances satisfy  condition \eqref{eqn:relation}. 

On the other hand  if 
  \begin{align*}    
      r_{13}&  = 12.129061710615553753, &   r_{24} & = 9.5117033174926140565, \\
       r_{12}  & = 8 & r_{14} & =  7.8020830551846857406, \\
      r_{23} & = 7.6549229903601603027, & r_{34} & =  7.3822682494734852600, 
  \end{align*}
we have $ m _1 / m _2 =0.69074480337446980353 $ and $ m _1 / m _4 = 0.87696321790891338292 $, also these distances satisfy condition \eqref{eqn:relation}. 

\end{remark}
We now show that, as a  consequence of Theorem \ref{thm:masses}, if  two of the masses are equal, there are  strong restrictions on the shape of the allowed central configuration. 
\begin{proposition} For any trapezoidal
c.c. in $\Omega $, if either $ m _2 = m _4 $ or $ m _1 = m _3 $  the configuration is a rhombus and the remaining two masses are necessarily equal. If 
$ m _3 = m _4 $, then the configuration is an isosceles trapezoid and the other two masses are necessarily equal.    \end{proposition}

\begin{proof}
  First we consider the case  $  m _1 = m _3 $. By the proof of Theorem \ref{thm:masses} we have that $ m _1 \geq m _3 $
  and the expression 
  \[
    \frac{ m _1 } { m _3 }      = \frac{ r _{ 14 } ^3 } { r _{ 23 } ^3 } \cdot \frac{ r _{ 12 } ^2 } { r _{ 34 } ^2 }\cdot \frac{ (r _{ 24 } ^3 - r _{ 34 } ^3)}{( r _{ 24 } ^3 - r _{ 14 } ^3)   } \cdot \frac{(r _{ 13 } ^3 - r _{ 23 } ^3) } { (r _{ 13 } ^3 - r _{ 12 } ^3)  } 
\]
is the product of three numbers greater than or equal to one. If  $ m _3 = m _1 $, it follows that
each of the fractions in the equation above is equal to one. From the first two fractions we obtain $ r _{ 14 } = r _{ 23 } $ and $ r _{ 12 } = r _{ 34 } $, respectively.  From the last two fractions we obtain $ r _{ 14 } = r _{ 34 } $ and $ r _{ 12 } = r _{ 23} $, respectively. This yields a rhombus. Moreover, comparing the two  equation \eqref{eqn:cc_n3} we find that $ m _2 = m _4 $. 

\

Second we consider  $  m _2 = m _4 $. From the proof of Theorem \ref{thm:masses}, we have that $ m _2 \geq m _4 $ and 
\begin{align*}
    \frac{ m _2 } {  m _4 }     =\left( \frac{ r _{ 12 } r _{ 23 } ^2 } { r _{ 34 } r _{ 14 } ^2 } \right)  ^2 \cdot  \frac{ r _{ 14 } } { r _{ 23 } } \cdot  \frac{(r _{ 13 } ^3 - r _{ 34 } ^3)}{(r _{ 13 } ^3 - r _{ 23 } ^3)} \cdot \frac{  (r _{ 24 } ^3 - r _{ 14 } ^3)    }{   (r _{ 24 } ^3 - r _{ 12 } ^3)    }.
\end{align*} 
This expression is  the product of four numbers greater than or equal to one. If we require $ m _2 = m _4 $, then from the second fraction we obtain $ r _{ 14 } = r _{ 23 }  $. Moreover, the first fraction now reduces to $ (r _{ 12 } / r _{ 34 } ) ^2 $, and thus we must have $ r _{ 12 } = r _{ 34 } $. The last two fractions 
 give $ r _{ 23 } = r _{ 34 } $ and $ r _{ 14 } = r _{ 12 } $, yielding again a  rhombus.

Third, suppose $ m _3 = m _4 $.  From equation \eqref{eqn:masses3}, we have that 
\[
    \frac{ m _3} { m _4 }  =  
    \frac{  r _{ 24 } ^{-3} - r _{ 14 } ^{-3} } { r _{ 13 } ^{-3} - r _{ 23 } ^{-3}} 
\]
If $ m _3 = m _4 $ this implies that $  r _{ 24 } ^{-3} - r _{ 14 } ^{-3} = r _{ 13 } ^{-3} - r _{ 23 } ^{-3} $. Consequently,
\[ r _{ 24 } ^{-3} - r _{ 13 } ^{-3} = r _{ 14 } ^{-3} - r _{ 23 } ^{-3}.\]
If $ r_{23} < r_{14} $, then  $r _{ 23 } ^{-3}>r _{ 14 } ^{-3} $ and  the right-hand side of equation above is negative, contradicting the fact that
$ r_{13}\geq r_{24} $. Hence, $ r _{ 23 } = r _{ 14 } $ and $ r _{ 13 } = r _{ 24 } $. Since a trapezoid with equal diagonals is an isosceles trapezoid, the configuration is an isosceles trapezoid. The equality of the masses $ m _1 = m _2 $ then follows from the first of equations \eqref{eqn:masses1}.  
This completes the proof. 

\end{proof}
\begin{remark} 
    Note that $ m _1 = m _2 $ does not imply that the configuration is an isosceles trapezoid. For instance,  a non-symmetric central configuration with $ m _1 = m _2 $ can be found numerically to be 
\begin{align*}    
      r_{13}&= 10.13318587483539368 , &   r_{24} & = 8.63262460668978253, \\
       r_{12}  & = 8 & r_{14} & = 7.59545875301365884, \\
      r_{23} & =7.03230033956929474, & r_{34} & = 4.37871386495945262.
  \end{align*}
  see Figure \eqref{fig:nonsymmetric}.

\end{remark} 
 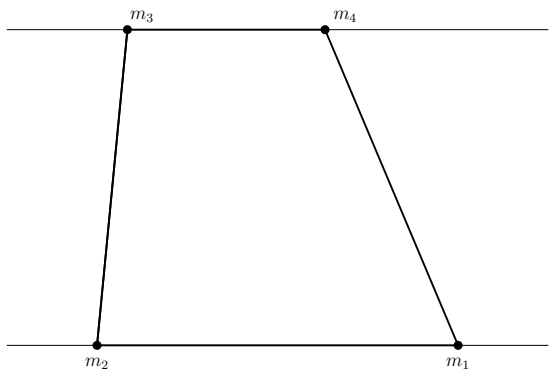
\begin{figure}
     \begin{center}
      \scalebox{.6}{
 \begin{tikzpicture}

\draw (-6,0) -- (6,0);
\draw (-6,7) -- (6,7);

\fill (4,0) circle (0.1cm);
\fill (-4,0) circle (0.1cm);
\fill (-3.3267,7) circle (0.1cm);
\fill (1.0519,7) circle (0.1cm);

\draw[very thick] (-4,0)--(4,0);
\draw[very thick] (-4,0)--(-3.32676299,7);
\draw[very thick] (-3.32676299,7)--(1.0519508707,7);
\draw[very thick] (-4,0)--(-3.32676299,7);
\draw[very thick] (1.0519508707,7)--(4,0);

\node[] at (4,-0.4) {$m _1 $};
\node[] at (-4,-0.4) {$m _2 $};
\node[] at (-3,7.3) {$m _3 $};
\node[] at (1.5,7.3) {$m _4 $};

\end{tikzpicture}
}\end{center}
\caption{An example of a non-symmetric trapezoidal   central configuration with $ m _1 = m _2 $.\label{fig:nonsymmetric}}
\end{figure}
 
\section*{Acknowledgments}
This work was supported by an NSERC discovery grant individual. The author is grateful to Shengda Hu for helpful discussions. 

\bibliographystyle{amsplain}
\bibliography{references}

\end{document}